\documentclass[letterpaper, 10 pt, conference]{IEEEtran}
\usepackage[letterpaper, total={6.8in, 8.7in}]{geometry}
\IEEEoverridecommandlockouts
\usepackage{cite}
\usepackage{amsmath,amssymb,amsfonts,amsthm}
\usepackage{algorithmic}
\usepackage{graphicx}
\usepackage{textcomp}
\usepackage{xcolor}
\def\BibTeX{{\rm B\kern-.05em{\sc i\kern-.025em b}\kern-.08em
    T\kern-.1667em\lower.7ex\hbox{E}\kern-.125emX}}

\usepackage{multirow}
\usepackage{graphics} 
\usepackage{epsfig} 
\usepackage{mathptmx} 
\usepackage{times} 
\usepackage[english]{babel}
\usepackage{epstopdf}
\usepackage{comment}
\usepackage{subcaption}
\newtheorem{theorem}{Theorem}[section]

\newtheorem{assumption}{Assumption}

\title{\LARGE \bf
Lyapunov Analysis of Least Squares Based Direct Adaptive Control}

\author{Nursefa Zengin, Baris Fidan, Ladan Khoshnevisan

\thanks{This is the second version of the arXiv:2007.08578 in Arxiv and it is submitted to IEEE CDC 2022. This work is supported by the Canadian NSERC Discovery Grant 116806."}

\thanks{Authors are with the Department of Mechanical and Mechatronics Engineering,
        University of Waterloo, ON, Canada
        {\tt\small nyarbasi,fidan,lkhoshne@uwaterloo.ca}}%

}

\begin{document}

\maketitle
\thispagestyle{empty}
\pagestyle{empty}

\begin{abstract}
Adaptive control strategies usually are designed based on gradient methods for the sake of simplicity in Lyapunov analysis. However, least squares (LS)-based parameter identifiers, with proper selection of design parameters, exhibit better transient performance than the gradient-based ones, from the aspects of convergence speed and robustness to measurement noise. On the other hand, most of the LS-based adaptive control procedures are designed via the indirect adaptive control approaches, due to the difficulty in integrating an LS-based adaptive law within the direct approaches starting with a certain Lyapunov-like cost function to be driven to (a neighborhood of) zero. In this paper, a formal constructive analysis framework is proposed to integrate the recursive LS-based parameter identification with direct adaptive control. To this end, a Lyapunov-like function is proposed for the analysis to achieve adaptive laws, which guarantee the exponential convergence of the parameters. Application of the proposed procedure in adaptive cruise control design is studied through Matlab/Simulink and CarSim simulations, validating the analytical results.
\end{abstract}


\section{Introduction}
 Stability and convergence analysis of adaptive controller schemes has traditionally been based on Lyapunov stability notions and techniques \cite{sun,fidan,narendrabook,goodwinbook,krsticbookadaptive} . Lyapunov-like functions are selected in the design of adaptive control schemes to penalize the magnitude of the tracking or regulation error but at the same time to facilitate designing an adaptive law to generate the parameter estimates used by the control law. Adaptive control designs targeting to drive a Lyapunov-like function to zero mostly lead to gradient based adaptive laws with constant adaptive gain. On the other hand, it is well observed that least-squares (LS) algorithms have the advantage of faster convergence; hence, LS based adaptive control has potential to enhance convergence performance in direct adaptive control approaches as well \cite{fidan,guler1,guler2,krstic,ls3}.

Despite wide use of gradient based online parameter identifiers, LS adaptive algorithms with forgetting factor are developed to be capable of faster settling and/or being less sensitive to measurement noises. Such properties have been justified by various simulation and experimental results \cite{guler1,guler2,ls_app4,slotine, koksal}. LS-based online parameter identification has been used for achieving better convergence and robustness to measurement noises in indirect adaptive control schemes as well as integrated direct/indirect adaptive control procedures  \cite{ls3,ls4,ls6,ls7,ls9,ls10,ls11,cho,krstic,karafyllis2018,jiang2015}.

In addition to the existing mathematical LS based adaptive control design studies, there are some publications in the recent literature on real-time applications, including those on robotic manipulators \cite{ls_app3,ls_app4,ls_app7, slotine}, unmanned aerial vehicles \cite{ls_app1,ls_app8,koksal}, and passenger vehicles \cite{vahidi,bae,pavkovic2009,acosta2016,chen2018,seyed2019,patra2015}. Most of the existing studies on LS based adaptive control follow the indirect approach as opposed to the direct adaptive control. One reason for this is that constructive Lyapunov analysis of direct adaptive control is complicated for producing an LS based adaptive control scheme.

This paper proposes a constructive analysis framework for recursive LS (RLS) online parameter identifier based direct model reference adaptive control (MRAC). In the literature, \cite{sun,fidan} considered the possible use of LS based online parameter identifiers in direct MRAC. However, the proof and the Lyapunov analysis were not provided in detail. Several techniques have been developed to robustify the LS based online parameter identifiers with respect to the loss of adaptation or parameter bursts related to the gain (covariance) matrix becoming arbitrarily small or arbitrarily large, including use of parameter projection, resetting, saturation, and forgetting factor \cite{ortega}. The role of forgetting factor is extensively investigated in \cite{goel} and \cite{guler1}, where it is demonstrated that without forgetting factor the parameter estimates converge to the real values only asymptotically (and typically slower), whereas  with forgetting factor the convergence becomes exponentially fast, which leads to specific design procedures for different applications. For instance, a composite LS method is provided in \cite{slotine}, where a Lyapunov-like function with time-varying gain matrix is utilized. However, the design in \cite{slotine} is for a specific system model suitable for robot manipulators, which limits the applicability of the proposed procedure as is.

Constructive Lyapunov analysis of RLS parameter identifier based direct adaptive control, which is used to build the adaptive control laws, is studied in this paper. The main difference from the gradient based approaches is replacement of the constant adaptation gain with a time varying adaptive gain (covariance) matrix. For a systematic construction of the direct MRAC scheme with time-varying adaptive gain, a Lyapunov-like function is constructed through which an LS parameter identification based direct adaptive control scheme is established to guarantee asymptotic stability. The proposed procedure is utilized in an adaptive cruise control (ACC) application case study to demonstrate the transient performance, validate the analytical results, and compare the performance with the gradient based adaptive controllers through a set of Matlab/Simulink and CarSim simulations.


The paper is organized as follows. Section II is dedicated to background on direct MRAC design. Section III provides Lyapunov-like function composition and analysis. Comparative simulation testing and analysis of the ACC application case study is presented in Section IV. Final remarks of the paper are given in Section V.

\section{Background: Direct Model Reference Adaptive Control}
In model reference adaptive control (MRAC), desired plant behaviour is described by a reference model which is often formulated in the form of a transfer function driven by a reference signal. Then, a control law is developed via model matching so that the closed loop system has a transfer function equal to the reference model \cite{sun,fidan,narendrabook,goodwinbook}.  Consider the SISO LTI plant
\begin{equation} \label{eq:xpdot}
\begin{aligned}
\dot{x}_p (t) &= A_p x_p (t) + B_p u_p (t), \quad x(0) = x_0, &\\ y_p(t)& = C_p^T x_p(t), &
\end{aligned}
\end{equation}
 \noindent
with state $x_p \in \mathbb{R}^n$, input $u_p \in \mathbb{R}$, output $y_p \in \mathbb{R}$, and system matrices $A_p, B_p, C_p$ of appropriate dimensions.
The transfer function of the plant is given by
\begin{equation} \label{eq:Gp}
G_p (s)=k_p \frac{Z_p (s)}{R_p (s)},
\end{equation}
\noindent
where $k_p$ is the high frequency gain, and $Z_p(s)$ and $R_p(s)$ are monic polynomials. Assume that the plant \eqref{eq:xpdot} is minimum phase, i.e, $Z_p(s)$ is Hurwitz. Consider the reference model
\begin{equation} \label{eq:xmdot}
\begin{aligned}
\dot{x}_m (t) &= A_m x_m (t) + B_m r (t), \quad x_m(0) = x_{m0}, &\\ y_m (t)& = C_m^T x_m (t), &
\end{aligned}
\end{equation}
\noindent
which is fed by the reference input signal $r \in \mathbb{R}$.
The transfer function of the reference model \eqref{eq:xmdot} is given by
\begin{equation} \label{eq:Wm}
 W_m (s)=k_m \frac{Z_m (s)}{R_m (s)},
\end{equation}
\noindent
where $k_m$ is the high frequency gain, and $Z_m(s)$ and $R_m(s)$ are monic polynomials. The MRAC task \cite{sun,fidan} is to generate the control signal $u_p$  so that all the closed-loop system signals are bounded and the plant output $y_p$ tracks the reference model output $y_m$, under the following assumptions:
\begin{assumption} {Plant Assumptions.}
\begin{itemize}
\item[i] $Z_p (s)$ is a monic Hurwitz polynomial.
\item[ii] Upper bound $n$ of the degree $n_p$ of $R_p (s)$ is known.
\item[iii] Relative degree $n^*=n_p-m_p$ of $G_p (s)$ is known, where $m_p$ denotes the degree of $Z_p (s)$.
\item[iv] $\textrm{sign}(k_p)$ is known.
\end{itemize}
\end{assumption}
\begin{assumption} {Reference Model Assumptions}
\begin{itemize}
\item[i] $Z_m (s), R_m (s)$ are monic Hurwitz polynomials of degree $q_m,p_m$, respectively.
\item[ii] Relative degree $n_m=p_m-q_m$ of $W_m (s)$ is the same as that of $G_p (s)$, i.e, $n^*=n^*_m$.
\end{itemize}
\end{assumption}
\noindent
Consider the fictitious feedback control law \cite{sun,fidan}
\begin{equation} \label{direct_u}
u_p={\theta^*_1}^T \frac{\alpha(s)}{\Lambda (s)} u_p+ {\theta^*_2}^T \frac{\alpha(s)}{\Lambda(s)} y_p+\theta^*_3 y_p+c^*_0 r,
\end{equation}
\noindent
where
\begin{equation*}
\begin{aligned}
& c^*_0=\frac{k_m}{k_p}, &\\ &\alpha (s) \triangleq \alpha_{n-2} (s) = [s^{n-2}, s^{n-3}, \cdots, s, 1]^T \quad \quad \text{for} \quad n\geq2, &\\ &\alpha (s) \triangleq 0 \qquad \qquad \qquad \qquad \qquad \qquad \qquad \quad  \text{for} \quad n = 1.  &
\end{aligned}
\end{equation*}
\noindent
$\Lambda (s)$ is an arbitrary monic Hurwitz polynomial of degree $n-1$ containing $Z_m (s)$ as a factor, i.e.,
\begin{equation*}
\Lambda (s) = \Lambda_0 (s) Z_m(s)
\end{equation*}
\noindent
implying that $\Lambda_0 (s)$ is monic and Hurwitz. The fictitious ideal model reference control (MRC) parameter vector $ \theta^*=\begin{bmatrix} \theta^{*T}_1 & \theta^{*T}_2 & \theta^*_3 & c^*_0 \end{bmatrix}^T$ is chosen so that the transfer function from $r$ to $y_p$ is equal to $W_m (s)$.
\noindent
The closed-loop reference to output relation for the MRC scheme above is derived in \cite{sun,fidan} as
\begin{equation} \label{gc}
y_p = G_c (s) r,
\end{equation}
\noindent
where
\begin{equation*}
G_c = \frac{c^*_0 k_p Z_p \Lambda^2}{\Lambda ( \Lambda- \theta^{*T}_1 \alpha R_p - k_p Z_p (\theta^{*T}_2 \alpha  + \theta^*_3 \Lambda)}.
\end{equation*}
\noindent
The ideal MRC parameter vector $\theta^*$ is selected to match the coefficients of $G_c (s)$ in \eqref{gc} and $W_m (s)$ in \eqref{eq:Wm}.
A state-space realization of the ideal MRC law \eqref{direct_u} is given by \cite{sun,fidan}
\begin{equation} \label{eq:omega_up_MRC}
\begin{aligned}
&\dot{\omega}_1(t)=F\omega_1(t)+g u_p(t), \quad \omega_1(0)=0, &\\& \dot{\omega}_2(t)=F\omega_2(t)+g y_p(t), \quad \omega_2(0)=0, &\\& u_p(t)=\theta^{*T} \omega(t),&
\end{aligned}
\end{equation}
\noindent
where $\omega_1, \omega_2 \in \mathbb{R}^{n-1},$
\begin{equation*}
\begin{aligned}
&\theta^*=\begin{bmatrix} \theta^{*T}_1 & \theta^{*T}_2 & \theta^*_3 & c^*_0 \end{bmatrix}^T, \quad \omega=\begin{bmatrix} \omega_1^T & \omega_2^T & y_p & r \end{bmatrix}^T, & \\ &F= \begin{bmatrix} -\lambda_{n-2} & -\lambda_{n-3} & -\lambda_{n-4} & \cdots & -\lambda_{0} \\  1 & 0 &0 & \cdots & 0 \\ 0 & 1 & 0 & \cdot & 0 \\  \vdots & \vdots & \ddots & \ddots & \vdots \\ 0 & 0 & \cdots & 1 & 0 \end{bmatrix}, &\\ & \Lambda (s)=s^{n-1}+\lambda_{n-2} s^{n-2} +\cdots+\lambda_1 s+\lambda_0=det(sI-F), &\\& g=\begin{bmatrix} 1 & 0 & \cdots&0 \end{bmatrix}^T.&
\end{aligned}
\end{equation*}
\noindent
The MRAC scheme for the actual case where the plant parameters are unknown is derived by following the certainty equivalence approach and modifying \eqref{eq:omega_up_MRC} as
\begin{equation} \label{eq:omega_up_MRAC}
\begin{aligned}
&\dot{\omega}_1(t)=F\omega_1+g u_p(t), \quad \omega_1(0)=0, &\\& \dot{\omega}_2(t)=F\omega_2(t)+g y_p(t), \quad \omega_2(0)=0, &\\& u_p(t)=\theta^{T}(t) \omega(t),&
\end{aligned}
\end{equation}
\noindent
where $\theta(t)$ is the online estimate of the unknown ideal MRC parameter vector $\theta^*$. The adaptive law to generate $\theta(t)$ can be formed considering the following composite state space representation of the closed-loop system \cite{fidan}:
\begin{equation} \label{eq:Ycdot}
\begin{aligned}
\dot{Y}_c(t) &= A_0 Y_c(t) + B_c u_p(t), &\\ y_p(t) &= C_c^T Y_c(t),
\end{aligned}
\end{equation}
\noindent
where $Y_c=[x_P^T,\omega_1^T, \omega_2^T]^T$,
\begin{equation*}
A_0 = \begin{bmatrix} A_p & 0 & 0 \\ 0 & F & 0 \\ g C_p^T & 0 & F  \end{bmatrix} , \quad B_c=\begin{bmatrix} B_p \\ g \\ 0 \end{bmatrix}, \quad C_c^T= \begin{bmatrix} C_p^T, & 0,& 0 \end{bmatrix}.
\end{equation*}
\noindent
The system equation \eqref{eq:Ycdot} can also be expressed as \cite{fidan}
\begin{equation}
\label{eq:Ycdot2}
\begin{aligned}
\dot{Y}_c(t) &= A_c Y_c(t) + B_c c_0^* r(t) + B_c ( u_p(t) - \theta^{*T} \omega (t) ),& \\ y_p(t) &= C_c^T Y_c(t), &
\end{aligned}
\end{equation}
\noindent
where
\begin{equation*}
A_c=\begin{bmatrix} A_p+B_p \theta_3^*C_p^T & B_p \theta_1^{*T} & B_p \theta_2^{*T} \\ g\theta^*_3C^T_p & F+g\theta^{*T}_1 & g\theta^{*T}_2 \\  g C^T_p & 0 & F \end{bmatrix}.
\end{equation*}
\noindent
Consider the fictitious system
\begin{equation}
\label{eq:Ymdot}
\begin{aligned}
\dot{Y}_m(t) &= A_c Y_m(t) + B_c c_0^* r(t),& \\ \bar{y}_m(t) &= C_c^T Y_m(t), &
\end{aligned}
\end{equation}
\noindent
which is obtained by substituting \eqref{eq:omega_up_MRC} in \eqref{eq:Ycdot2}. Noting that the ideal MRC law \eqref{eq:omega_up_MRC} guarantees the matching of the closed-loop $r$-to-$y_p$ transfer function with the reference model transfer function $W_m(s)$, the $r$-to-$\bar{y}_m$ transfer function of \eqref{eq:Ymdot} satisfies
\begin{equation} \label{eq:Wm2}
W_m(s) = C_c^T (sI-A_c)^{-1} B_c c_0^*.
\end{equation}
\noindent
Furthermore, since \eqref{eq:xmdot} and \eqref{eq:Ymdot} are state-space representations of the same stable transfer function $W_m(s)$, $y_m(t)-\bar{y}_m(t)$ converges to zero exponentially fast. Hence, defining the output tracking error $e_1=y_p-\bar{y}_m$, exponential convergence of the tracking error $y_p-y_m$ to zero is equivalent to exponential convergence of $e_1$ to zero. Moreover, defining the state mismatch vector $e=Y_c-Y_m \in \mathbb{R}^{3n-2}$ and subtracting \eqref{eq:Ymdot} from \eqref{eq:Ycdot2} we obtain
\begin{equation} \label{eq:edot}
\begin{aligned}
&\dot{e}(t)=A_c e(t)+B_c (u_p(t)-\theta^{*T} \omega(t)), \quad e(0)=e_0, &\\& e_1=C^T_c e. &
\end{aligned}
\end{equation}
\noindent
Hence, we have
\begin{equation}
e_1 = W_m(s) \rho^* ( u_p - \theta^{*T} \omega ),
\end{equation}
\noindent
where $\rho^*=1/c_0^*$. Further, substituting \eqref{eq:omega_up_MRAC} into \eqref{eq:edot}, we obtain
\begin{equation} \label{eq:edot2}
\begin{aligned}
&\dot{e}=A_c e+B_c \tilde{\theta}^{T} \omega, &\\& e_1=C^T_c e, &
\end{aligned}
\end{equation}
\noindent
where
\begin{equation}
\tilde{\theta}(t) = \theta(t)-\theta^*.
\end{equation}

\section{Lyapunov-Like Function Composition and Analysis for Least-Squares Based Direct MRAC}
\noindent
In the typical direct adaptive control designs of the literature, which are gradient adaptive law based, the Lyapunov-like function is chosen as

\begin{equation} \label{eq:V1}
V_1(\tilde{\theta}, e)=\frac{e^T P_c e}{2}+\frac{\tilde{\theta}^T \Gamma^{-1} \tilde{\theta}}{2}|\rho^*|,
\end{equation}

\noindent
where $\tilde{\theta}=\theta-\theta^*$, $\theta^*$ and $\theta$, respectively, are the ideal MRC and actual MRAC parameter vectors defined in Section II, $P_c=P^T_c$ is a positive definite matrix satisfying certain conditions to be detailed in the sequel, and $\Gamma=\Gamma^T$ is a constant positive definite adaptive gain matrix. $P_c$ is selected to satisfy the Meyer-Kalman-Yakubovich Lemma \cite{fidan} algebraic equations

\begin{equation} \label{eq:PcAc}
\begin{aligned}
 P_c A_c+ A_c^T P_c&=-qq^T-\nu_c L_c, &\\P_c B_c c^*_0&=C_c, &
\end{aligned}
\end{equation}
\noindent
where $q$ is a vector, $L_c=L^T_c >0$, and $\nu_c>0$ is small constant. The time derivative $\dot{V}_1$ of $V_1$ along \eqref{eq:edot2},\eqref{eq:PcAc} is
\begin{equation}
\label{eq:V1dot}
\dot{V}_1=-\frac{e^T qq^T e}{2}-\frac{\nu_c}{2}e^T L_c e+e^T P_c B_c c^*_0\rho^* \tilde{\theta}^T \omega+\tilde{\theta}^T \Gamma^{-1} \dot{\tilde{\theta}}|\rho^*|.
\end{equation}
\noindent
Since $e^T P_c B_c c^*_0=e^T C_c=e_1$ and $\rho^*=|\rho^*| sgn(\rho^*)$, defining the gradient based adaptive law
\begin{equation} \label{eq:thetadot}
\dot{\theta}=-\Gamma \omega e_1 sgn(\rho^*)
\end{equation}
\noindent
leads to
\begin{equation} \label{eq:V1dot2}
\dot{V}_1=-\frac{e^T qq^T e}{2}-\frac{\nu_c}{2}e^T L_c e \leq 0,
\end{equation}
\noindent
noting that $\dot{\tilde{\theta}}=\dot{\theta}$. The equations \eqref{eq:V1} and \eqref{eq:V1dot2} imply that $V, \tilde{\theta}, \theta \in \mathcal{L}_\infty$ and $e \in \mathcal{L}_2 \cap \mathcal{L}_\infty$. Further,since the reference model system \eqref{eq:Wm2} is stable, we have $Y_m \in \mathcal{L}_\infty$. Hence we also have $Y_c = Y_m +e$, $\bar{y}_m = C_c^T Y_m\in \mathcal{L}_\infty$. By \eqref{eq:xpdot},\eqref{eq:omega_up_MRAC},\eqref{eq:Ycdot2}, this further implies that $x_p, y_p, \omega_1, \omega_2, u_p = \theta^T\omega \in \mathcal{L}_\infty$, i.e., all the signals in the closed-loop plant are bounded. Moreover, since $\dot{e} \in \mathcal{L}_\infty$, based on Barbalat's Lemma, $\lim_{t \to +\infty} e(t)=0$. Hence, the tracking error $e_1 = y_p - \bar{y}_m = C_c^T e$  converges to zero as time goes to infinity.

With the gradient based adaptive law \eqref{eq:thetadot} with constant adaptive $\Gamma$ gain, fast adaptation can be achieved only by using a large adaptive gain to reduce the tracking error rapidly. However, introduction of a large adaptive gain $\Gamma$ in many cases leads to high-frequency oscillations which adversely affects robustness of the adaptive control law.

Unlike the gradient based adaptive law \eqref{eq:thetadot} with constant adaptive gain $\Gamma$, generation of a time varying adaptive law gain matrix $P(t)$ that is adjusted based on identification error during estimation process, would allow an initial large adaptive gain to be set arbitrarily and then to be driven to  lower values to adaptively achieve the desired tracking performance.

For generation of the time varying gain $P(t)$, an efficient systematic approach is use of LS based adaptive laws, which are observed to have the advantage of faster convergence and robustness to measurement noises \cite{fidan,guler1,guler2,krstic,ls3}. Next, we propose a formal constructive analysis framework for integration of recursive LS (RLS) based parameter identification to direct adaptive control, following the steps above, but constructing a new Lyapunov-like function to replace \eqref{eq:V1}, aiming to formally establish am adaptive control scheme law that involves the control structure \eqref{eq:omega_up_MRAC} and an RLS based alternative of the adaptive law \eqref{eq:thetadot}.

Aiming to replace the constant adaptive gain $\Gamma$ with a time-varying  gain matrix $P(t)$, consider the following Lyapunov-like function in place of \eqref{eq:V1}:
\begin{equation} \label{eq:V2}
V_2(\tilde{\theta}, e, t)=\frac{e^T(t) P_c e(t)}{2}+\frac{\tilde{\theta}^T(t) P^{-1}(t) \tilde{\theta}(t)}{2}|\rho^*|,
\end{equation}
where $P(t)$ is uniformly positive definite. The time derivative $\dot{V}_2$ of $V_2$ along the solution of \eqref{eq:V2} is
\begin{equation} \label{eq:V2dot}
\begin{aligned}
\dot{V}_2=&-\frac{e^T qq^T e}{2}-\frac{\nu_c}{2}e^T L_c e+e^T P_c B_c c^*_0\rho^* \tilde{\theta}^T \omega &\\ &+\frac{1}{2} \tilde{\theta}^T \frac{d(P^{-1})}{dt} \tilde{\theta}|\rho^*| + \tilde{\theta}^T P^{-1} \dot{\tilde{\theta}} |\rho^*|,&
\end{aligned}
\end{equation}
where
\begin{equation} \label{fraction}
\frac{d(P^{-1})}{dt} = -P^{-1} \dot{P} P^{-1}.
\end{equation}
If $P(t)$ is updated according to the RLS adaptive law with forgetting factor,
\begin{equation} \label{eq:Pdot}
\dot{P}=\beta P - P \omega \omega^T P,
\end{equation}
where $\beta>0$ is the forgetting factor (scalar design coefficient), then \eqref{fraction} becomes
\begin{equation} \label{fraction2}
\frac{d(P^{-1})}{dt} = -\beta P^{-1} + \omega \omega^T.
\end{equation}
Substituting \eqref{fraction2} into \eqref{eq:V2dot}, we obtain
\begin{equation} \label{eq:V2dot2}
\begin{aligned}
\dot{V}_2=&-\frac{e^T qq^T e}{2}-\frac{\nu_c}{2}e^T L_c e +e_1 \rho^* \tilde{\theta}^T \omega - \frac{\beta}{2} \tilde{\theta}^T P^{-1} \tilde{\theta}|\rho^*| &\\&+ \tilde{\theta}^T P^{-1} \dot{\tilde{\theta}} |\rho^*| + \frac{1}{2} \omega^T \tilde{\theta} \tilde{\theta}^T \omega |\rho^*|.&
\end{aligned}
\end{equation}
Defining the adaptive law
\begin{equation} \label{eq:thetadot2}
\dot{\theta}=-P \omega e_1  sgn(\rho^*) -\frac{1}{2} P \omega^T \tilde{\theta} \omega,
\end{equation}
where $P(t)$ is updated via \eqref{eq:Pdot}, noting that $\dot{\tilde{\theta}}=\dot{\theta}$, and substituting into \eqref{eq:V2dot2}, we obtain
\begin{equation} \label{eq:V2dot3}
\dot{V}_2=-\frac{1}{2}e^T qq^T e-\frac{\nu_c}{2}e^T L_c e-\frac{1}{2}\tilde{\theta}^T P^{-1} \beta \tilde{\theta} |\rho^*| \leq 0,
\end{equation}
leading to the following theorem, which summarizes the  stability properties of the LS based direct MRAC scheme  \eqref{eq:omega_up_MRAC},\eqref{eq:Pdot},\eqref{eq:thetadot2}.

\begin{theorem}
The RLS parameter estimation based MRAC scheme \eqref{eq:omega_up_MRAC},\eqref{eq:Pdot},\eqref{eq:thetadot2} has the following properties:
\begin{itemize}
\item[i.] All signals in the closed-loop are bounded and tracking error converges to zero in time for any reference input $r \in \mathcal{L}_\infty$.
\item[ii.] If the reference input $r$ is sufficiently rich of order $2n$, $\dot{r}\in \mathcal{L}_\infty$, and $Z_p(s), R_p(s)$ are relatively coprime, then $\omega$ is persistently exciting (PE), viz.,

\begin{equation} \label{PE}
\int^{t+T_0}_{t} \omega(\tau) \omega^T(\tau) d\tau \geq \alpha_0T_0 I, \quad \alpha_0, T_0 >0, \quad \forall t\geq 0,
\end{equation}

which implies that $P,P^{-1} \in \mathcal{L}_{\infty}$ and $\theta(t)\to\theta^{*}$ as $t\to\infty$. In the case of $\beta=0$ (pure-RLS), $\theta\rightarrow\theta^*$ and $e_1\rightarrow 0$ as $t\rightarrow \infty$. When $\beta>0$, which is RLS with forgetting factor,  the parameter error $\| \tilde{\theta} \| = \| \theta-\theta^* \|$ and the tracking error $e_1$ converges to zero exponentially fast.
\end{itemize}
\end{theorem}

\begin{proof}
\noindent
\begin{itemize}
\item[i.] $e \in \mathcal{L}_2$, $\theta,\omega, \in \mathcal{L}_\infty$, and $\dot{e} \in \mathcal{L}_\infty$. Therefore, all signals in the closed loop plant are bounded. In order to complete the design, we need to show tracking error $e_1$ converges to the zero asymptotically with time. Using \eqref{eq:V2}, \eqref{eq:V2dot3}, we know that $e, e_1 \in \mathcal{L}_2$. Using, $\theta,\omega,e \in \mathcal{L}_\infty$ in \eqref{eq:edot2}, we have $\dot{e},\dot{e}_1 \in \mathcal{L}_\infty$. Since  $\dot{e},\dot{e}_1 \in \mathcal{L}_\infty$ and  $e_1 \in \mathcal{L}_2$, the tracking error $e_1$ goes to zero as $t$ goes to infinity.
\item[ii.]  Considering pure-RLS, when $\beta=0$, from \eqref{eq:Pdot} we have $\dot{P}=-P \omega \omega^T P\leq 0$. So, $P(t)$ is non-increasing i.e.,  $P(t)\leq P_0$. As $P(t)=P^T(t)> 0$, it has a limit i.e., $\lim_{t\to\infty} {P(t)}=\bar{P}$, where $\bar{P}=\bar{P}^T$ is a positive constant matrix. Furthermore, \eqref{eq:Pdot} results in
\begin{eqnarray*}
P(t)&=&P_0- \int_{0}^{t} P(\tau)\omega \omega^T P(\tau) \,d\tau \\
&\leq&\lambda_{min}(P_0)-\lambda_{min}(P_0) \int_{t-T_0}^{t} \omega \omega^T \,d\tau,
\end{eqnarray*}
where $\lambda_{min}(P_0)$ is the minimum singular value of $P_0$. By Theorem 3.4.3 of \cite{fidan}, if $r$ is sufficiently rich of order $2n$ then the $2n$ dimensional regressor vector $\omega$ is PE.  Because $\omega(t)$ is PE, we have
\begin{equation}
P(t) \leq \lambda_{min}(P_0)(1-\alpha T_0)I.
\end{equation}
As a result,
\begin{equation}
\bar{P} \leq P(t) \leq \gamma_1 I.
\end{equation}
So, $P(t)\in\mathcal{L}_\infty$ in pure-LS method. On the other hand, (\eqref{eq:V2dot3}) results in $\dot{V}_2(t)\leq 0$. Therefore, $V, \tilde{\theta} \in \mathcal{L}_\infty$ and $e \in \mathcal{L}_2 \cap \mathcal{L}_\infty$. Since $e = Y_c - Y_m$ and $\bar{y}_m(t) = C_c^T Y_m(t)$, $Y_m, \bar{y}_m, Y_c \in \mathcal{L}_\infty$ that gives use $ x_p, y_p, \omega_1, \omega_2 \in \mathcal{L}_\infty$. Moreover, we have $u_p = \theta^T\omega$ and $\theta, \omega \in \mathcal{L}_\infty$; therefore, $u_p \in \mathcal{L}_\infty$. So, all the signals in the closed-loop plant are bounded. Moreover, since $\dot{e} \in \mathcal{L}_\infty$, based on Barbalat's Lemma, $\lim_{t \to +\infty} e(t)=0$. Hence, the tracking error $e_1 = y_p - \bar{y}_m = C_c^T e$ converges to zero as time goes to infinity.

Now, we consider RLS with forgetting factor, in which $\beta>0$. Let $Q=P^{-1}$ and \eqref{fraction2} can be rewritten as
\begin{equation} \label{fraction2_new}
\dot{Q} = -\beta Q + \omega \omega^T.
\end{equation}
and the solution becomes
\begin{equation} \label{Q_integral}
Q(t) = e^{-\beta t} Q_0 + \int^t_0 e^{-\beta (t-\tau)} \omega(\tau) \omega^T(\tau) d\tau.
\end{equation}
Since $\omega(t)$ is PE,
\begin{equation} \label{Q_integral2}
\begin{aligned}
Q(t) &\geq \int^t_{T_0} e^{-\beta (t-\tau)} \omega(\tau) \omega^T(\tau) d\tau &\\&\geq \bar{\alpha}_0 e^{-\beta T_0} \int^t_{T_0} e^{-\beta (t-\tau)} \omega(\tau) \omega^T(\tau) d\tau &\\&\geq \beta_1 e^{-\beta T_0} I, \qquad \forall t \geq T_0,&
\end{aligned}
\end{equation}
where $\beta_1=\bar{\alpha}_0\alpha_0T_0,$ and  $\alpha_0,\bar{\alpha}_0,T_0>0$ are design constants, given in \eqref{PE}. For $t\leq T_0$,
\begin{equation} \label{c12}
Q(t) \geq e^{-\beta T_0}Q_0 \geq \lambda_{min}(Q_0)e^{-\beta T_0}I \geq \gamma_2 I \quad \forall t \geq 0,
\end{equation}
where $\gamma_2 = min\{\frac{\alpha_0 T_0}{\beta},\lambda_{min}(Q_0)\}e^{-\beta T_0}$. Since $\omega$ is PE,
\begin{equation} \label{c13}
Q(t) \leq Q_0 +\beta_2\int^t_{0} e^{-\beta (t-\tau)}d\tau I \leq \gamma_3 I, \quad \beta_2>0.
\end{equation}
where $\gamma_3 = \lambda_{max}(Q_0)+\frac{\beta_2}{\beta}>0$. Using \eqref{c12} and \eqref{c13}, we obtain
\begin{equation}
\gamma^{-1}_3 I \leq P^{-1}(t)=Q((t)\leq \gamma^{-1}_2 I.
\end{equation}
Therefore, $P(t),Q(t)\in\mathcal{L}_\infty$.
Based on \eqref{eq:V2dot3} and \eqref{eq:PcAc}, we have

\begin{equation}
\dot{V}_2(t) \leq -\lambda_{min}(A_c) e^T P_c e-\frac{1}{2} \beta \tilde{\theta} P^{-1} \tilde{\theta} |\rho^*|.
\end{equation}

where $\lambda_{min}(A_c)$ is the singular value of $A_c$. Therefore, considering $\alpha_0=min(\lambda_{min}(A_c),\frac{1}{2}\beta)$ results in

\begin{equation}
\dot{V}_2(t) \leq -\alpha_0 (e^T P_c e-\tilde{\theta} P^{-1} \tilde{\theta}|\rho^*|)=-\alpha_0 V_2(t).
\end{equation}

So, we have
\begin{equation}
V_2(t) \leq e^{-\alpha_0 t} V_2(0).
\end{equation}
which implies exponential convergence of $\tilde{\theta}$ and the tracking error $e_1(t)$. Exponential convergence is interesting from the adaptive control point of view, as it provides fast adaptation and robustness against noise and external disturbances, which are inevitable in practical applications.

\end{itemize}

\end{proof}

\noindent
Comparing two adaptive laws in \eqref{eq:thetadot} and \eqref{eq:thetadot2}, we can clearly see the effect of time varying covariance matrix reflected as an additional term to the similar part of \eqref{eq:thetadot}.

\section{A Simulation Case Study}
For the application of RLS based adaptive control, ACC case study is considered, and all three procedures, gradient method, pure-RLS (PRLS) ($\beta=0$) and RLS ($\beta>0$) are simulated to scrutinize the efficacy of RLS method. A basic ACC scheme is given in Fig. \ref{acc_scheme}. ACC regulates the following vehicle's speed $v$ towards the leading vehicle's speed $v_l$ and keeps the distance between vehicles $x_r$ close to desired spacing $s_d$.
\begin{figure}[h!]
\centering
	\includegraphics[width=0.4\textwidth , height=0.1 \textwidth]{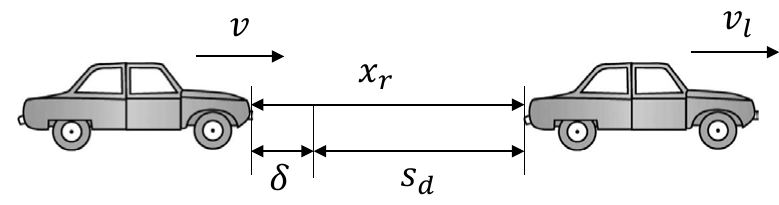}
	\caption{Leading and following vehicles.}
\label{acc_scheme}
\end{figure}
\noindent
The control objective in ACC is to make the speed error close to zero as time increases. This objective can be expressed as
\begin{equation}
v_r \rightarrow 0, \quad \delta \rightarrow 0, \quad t\rightarrow \infty,
\end{equation}
\noindent
where $v_r =v_l-v$ which is defined as the speed error or sometimes relative speed, $\delta=x_r-s_d$ is the spacing error. The desired spacing is proportional to the speed since the desired spacing between vehicles is given as
\begin{equation}
s_d=s_0+h v
\end{equation}
\noindent
where $s_0$ is the fixed spacing for safety so that the vehicles are not touching each other at zero speed and $h$ is constant time headway. Control objective should also satisfies that $a_{min} \leq \dot{v} \leq a_{max}$, and small $|\ddot{v}|$. First constraint restricts ACC vehicle generating high acceleration and the second one is given for the driver's comfort. For ACC system, a simple model is considered approximating the actual vehicle longitudinal model without considering nonlinear dynamics which is given by
\begin{equation}
\dot{v}=-a v+b u+d,
\end{equation}
where $v$ is the longitudinal speed, $u$ is the throttle/brake command, $d$ is the modeling uncertainty, $a$ and $b$ are positive constant parameters. We assume that $d,\dot{d}v_l,\dot{v}_l$ are all bounded. MRAC is considered so that the throttle command $u$ forces the vehicle speed to follow the output of the reference model
\begin{equation}
v_m=\frac{a_m}{s+a_m} (v_l+k\delta),
\end{equation}
where $a_m$ and $k$ are positive design parameters.We first assume that $a,b,$ and $d$ are known and consider the control law as follows:
\begin{equation}
u=k^*_1 v_r+k^*_2 \delta +k^*_3,
\end{equation}
\begin{equation}
k^*_1=\frac{a_m-a}{b}, \quad k^*_2=\frac{a_m k}{b}, \quad k^*_3=\frac{a v_l-d}{b}.
\end{equation}
Since $a,b,$ and $d$ are unknown, we change the control law as
\begin{equation} \label{u_acc}
u=k_1 v_r+k_2 \delta +k_3,
\end{equation}
where $k_i$ is the estimate of $k^*_i$ to be generated by the adaptive law so that the closed-loop stability is guaranteed. The tracking error is given as
\begin{equation} \label{acc_bdpm}
e=v-v_m=\frac{b}{s+a_m}(-k^*_1 v_r-k^*_2 \delta -k^*_3+u).
\end{equation}
Substituting the control law in \eqref{u_acc} into \eqref{acc_bdpm}, we obtain
\begin{equation} \label{acc_error}
e=\frac{b}{s+a_m}(\tilde{k}_1 v_r+\tilde{k}_2 \delta +\tilde{k}_3),
\end{equation}
where $\tilde{k}_i=k_i-k^*_i$ for $i=1,2,3$. In order to find the adaptive law, consider the Lyapunov function and its time derivative \cite{fidan} as
\begin{equation} \label{acc_lyapunov}
V=\frac{e^2}{2}+\sum ^{3}_{i=1} \frac{b}{2 \gamma _i} \tilde{k}^2_i \quad \gamma_i>0, b>0,
\end{equation}
\begin{equation}
\dot{V}=-a_m e^2+b e (\tilde{k}_1 v_r+\tilde{k}_2 \delta +\tilde{k}_3) + \sum ^{3}_{i=1} \frac{b}{\gamma _i} \tilde{k}_i \dot{\tilde{k}}_i .
\end{equation}
Therefore, the following gradient based adaptive laws are applied to ACC
\begin{equation}
\begin{aligned}
\dot{k}_1&=Pr\{-\gamma_1 e v_r\}, &\\ \dot{k}_2&=Pr\{-\gamma_2 e \delta\}, &\\ \dot{k}_3&=Pr\{-\gamma_3 e\},&
\end{aligned}
\end{equation}
where the projection operator keeps $k_i$ within the lower and upper intervals and $\gamma_i$ are the positive constant adaptive gains. These adaptive laws lead to
\begin{equation}
\dot{V}=-a_m e^2.
\end{equation}
By projection operator, estimated parameters are guaranteed to be bounded by forcing them to remain inside the bounded sets, $\dot{V}$ implies that $e\in\mathcal{L}_\infty$, in turn all other signals in the closed loop are bounded. We apply RLS based adaptive law to \eqref{acc_lyapunov} and obtain following equations to be used in simulations
\begin{figure} [h!]
\centering
	\includegraphics[width=0.54 \textwidth , height=0.46\textwidth]{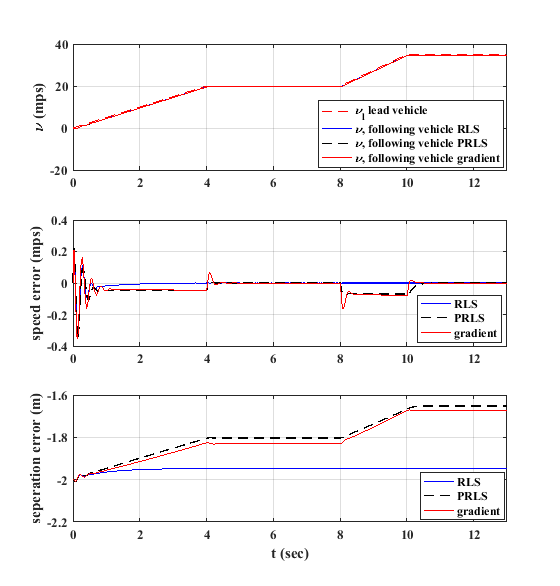}
	\caption{ACC comparison results, speed tracking and separation error in Matlab/Simulink.}
	\label{acc3}
\end{figure}
\begin{equation} \label{rls_dmracacc}
\begin{aligned}
\dot{\hat{\theta}}&=Pr\{\phi P_{ii} e \},  &\\ \dot{P}&= \beta P- P\phi \phi^T P, \quad&
\end{aligned}
\end{equation}
with
\begin{equation}
\begin{aligned}
&e=v-v_m, & \\ &\theta=\begin{bmatrix} k_1, & k_2, & k_3 \end{bmatrix} ^T, &\\ &\phi= \begin{bmatrix} \frac{ v_r}{s+a_m}, & \frac{ \delta}{s+a_m}, & \frac{1}{s+a_m} \end{bmatrix}^T, &
\end{aligned}
\end{equation}
where $P_{ii}$ are the diagonal elements of $P$ covariance matrix, i=1,2,3.

For gradient based adaptive law, $\gamma_1=50, \gamma_2=30, \gamma_3=40$ constant gains are given. For RLS based algorithm $\beta=0.95$ and $P(0)=100I_3$ are given. For both RLS and gradient schemes, a Gaussian noise is applied ($\sigma = 0.05$). Simulation results from Matlab/Simulink for throttle system are given in Fig. \ref{acc3}. Fig. \ref{acc3} shows the vehicle following for both gradient based adaptive law and RLS based adaptive law. The speed error in velocity tracking shows the better performance for RLS adaptive law. Furthermore, it can be inferred from this figure that the RLS provides exponential convergence.

\begin{figure} [h!]
\centering
	\includegraphics[width=0.54 \textwidth , height=0.46\textwidth]{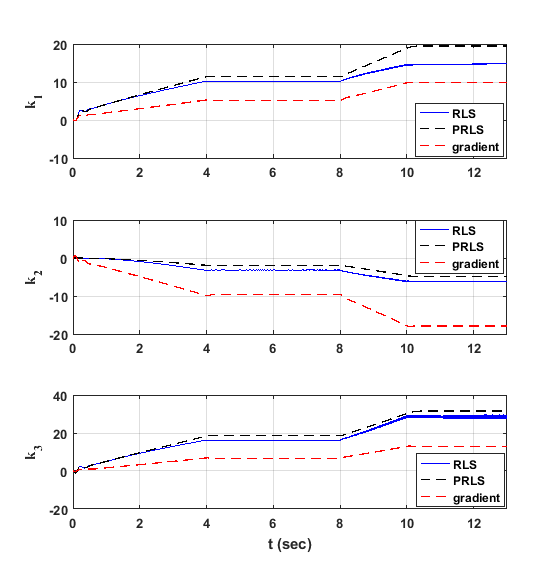}
	\caption{ACC adaptive parameters in Matlab/Simulink.}
	\label{acc4}
\end{figure}

The adaptive parameters are illustrated in Fig. \ref{acc4}. It is shown that these parameters are achieved adaptively based on the system specifications and they are all stable.

\begin{figure} [h!]
\centering
	\includegraphics[width=0.54 \textwidth , height=0.44\textwidth]{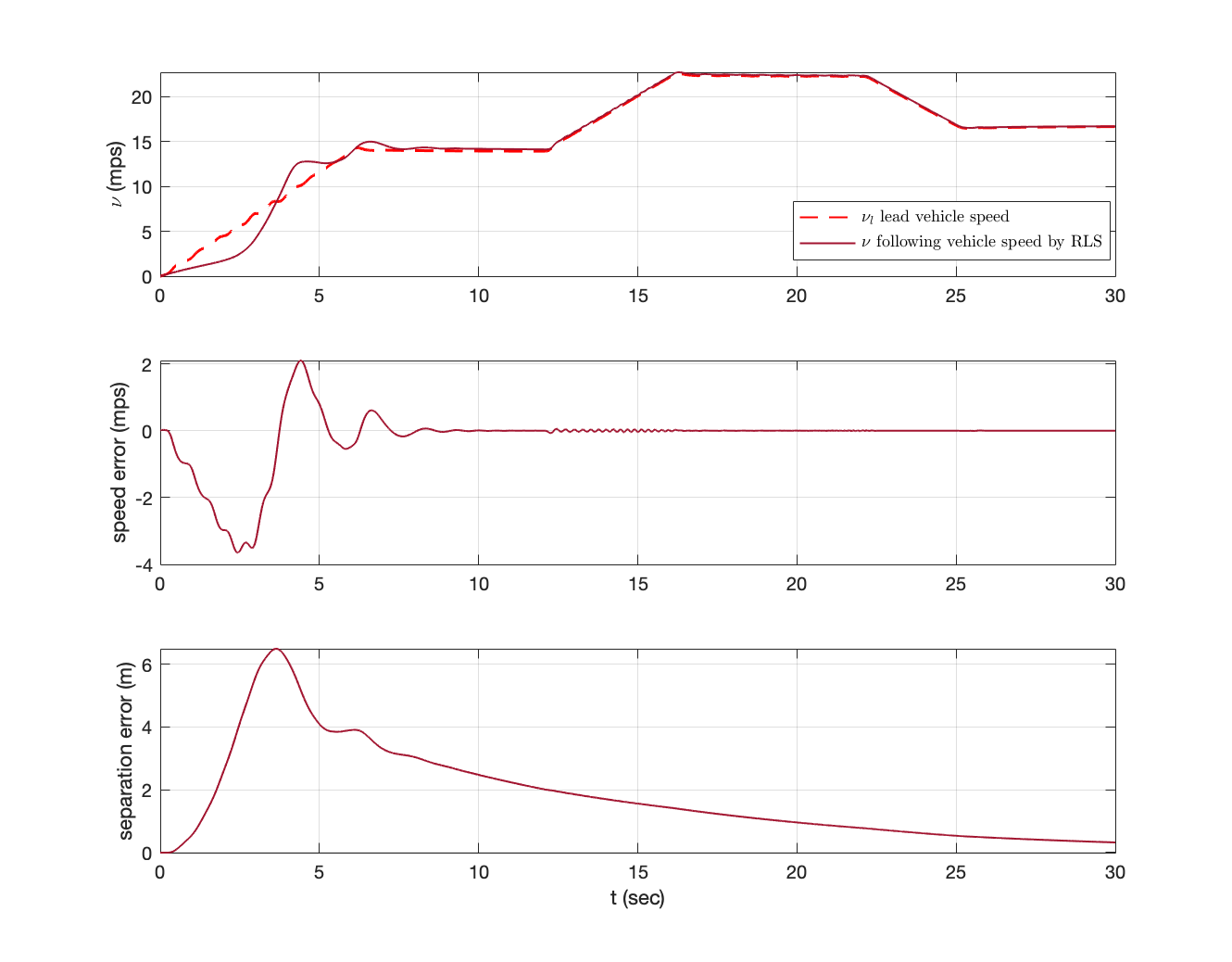}
	\caption{RLS based ACC results in CarSim.}
	\label{acc5}
\end{figure}

We also implemented RLS based adaptive control algorithm in \eqref{rls_dmracacc} to CarSim for more realistic results. The vehicle parameters used in CarSim are as follows: $m=567.75~ kg$, $R=0.3~ m$, $I=1.7 ~kgm^2$, $B=0.01 ~kg/s$. The adaptive gains for both gradient and RLS are used the same as in Matlab/Simulink. CarSim results for RLS based ACC can be found in Fig. \ref{acc5}. Results demonstrate the ability of the following vehicle equipped with RLS based adaptive law on dry road by adjusting the speed and the distance between the leading vehicle and itself.

\section{CONCLUSIONS}
In this paper, a constructive Lyapunov analysis of RLS based parameter estimation direct adaptive control is proposed. A systematic representation of designing a time varying adaptive gain (covariance) matrix is proposed via Lyapunov analysis, where it is analytically demonstrated that the adaptive parameters converge to the actual ones exponentially fast. The simulation results on an ACC model via Matlab/Simulink, and the comparative achievements, with respect to gradient based method and pure-LS procedure, validate the analytical gains. Moreover, it is shown that LS-based approach outperforms the others. Furthermore, a realistic vehicle software, CarSim, is utilized to scrutinize the applicability of the proposed method on a real ACC system.

\bibliographystyle{IEEEtran}
\bibliography{ZFK_CDC2022_LS_DAC_07}

\end{document}